\documentclass[english,a4paper]{article}

\usepackage{babel}
\usepackage[utf8]{inputenc}
\usepackage[pdftex]{hyperref}

\usepackage{color}
\usepackage{graphicx}
\usepackage{relsize}
\usepackage{enumitem}
\usepackage{pgf}

\usepackage{amsmath}
\usepackage{amsthm}
\usepackage{amssymb}

\allowdisplaybreaks

\newcommand{\q}[1]{``#1''}

\newcommand{\ket}[1]{\left|\, #1 \, \right\rangle}

\newcommand{\e}{\text{e}}
\newcommand{\imag}{i}

\newcommand{\ZZ}{\mathbb{Z}}

\newcommand{\hidepart}[1]{}

\newtheorem{theorem}{Theorem}
\newtheorem{lemma}{Lemma}
\newtheorem{claim}{Claim}
\newtheorem{definition}{Definition}

\newtheoremstyle{TheoremNum}
        {\topsep}{\topsep}
        {\itshape}
        {}
        {\bfseries}
        {.}
        { }
        {\thmname{#1}\thmnote{ \bfseries #3}}
\theoremstyle{TheoremNum}

\newcommand{\footremember}[2]{%
   \footnote{#2}
    \newcounter{#1}
    \setcounter{#1}{\value{footnote}}%
}
\newcommand{\footrecall}[1]{%
    \footnotemark[\value{#1}]%
}

\title{Quantum algorithms for computing short discrete logarithms and factoring RSA integers}
\author{Martin Ekerå\footremember{kth}{KTH Royal Institute of Technology, SE-100 44 Stockholm, Sweden.}$^,\,$\footremember{ncsa}{Swedish NCSA, Swedish Armed Forces, SE-107 85 Stockholm, Sweden.} \and Johan Håstad\footrecall{kth}}

\begin{document}
\maketitle

\begin{abstract}
  In this paper we generalize the quantum algorithm for computing short discrete
  logarithms previously introduced by Ekerå \cite{ekera} so as to allow for
  various tradeoffs between the number of times that the algorithm need be executed
  on the one hand,
  and the complexity of the algorithm and the requirements it imposes on the
  quantum computer on the other hand.

  Furthermore, we describe applications of algorithms for computing short
  discrete logarithms. In particular, we show how other important problems such as
  those of factoring RSA integers and of finding the order of groups under
  side information may be recast as short discrete logarithm problems.
  This immediately gives rise to an algorithm for factoring RSA integers that
  is less complex than Shor's general factoring algorithm in the sense that it
  imposes smaller requirements on the quantum computer.

  In both our algorithm and Shor's algorithm, the main hurdle is to compute
  a modular exponentiation in superposition. When factoring an $n$ bit integer,
  the exponent is of length $2n$ bits in Shor's algorithm, compared to
  slightly more than $n/2$ bits in our algorithm.
\end{abstract}

\section{Introduction}
In a groundbreaking paper \cite{shor1994} from 1994, subsequently extended
and revised in a later publication \cite{shor1997}, Shor introduced
polynomial time quantum computer algorithms for factoring integers over
$\mathbb Z$ and for computing discrete logarithms in the multiplicative group
$\mathbb F_p^*$ of the finite field $\mathbb F_p$.

Although Shor's algorithm for computing discrete logarithms was originally
described for $\mathbb F_p^*$, it may be generalized to any finite cyclic
group, provided the group operation may be implemented efficiently using
quantum circuits.

\subsection{Recent work}
Ekerå \cite{ekera} has introduced a modified version of Shor's algorithm for
computing short discrete logarithms in finite cyclic groups.

Unlike Shor's original algorithm, this modified algorithm does not require the
order of the group to be known. It only requires the logarithm to be short; i.e.
it requires the logarithm to be small in relation to the group order.

The modified algorithm is less complex than Shor's general algorithm when the
logarithm is short. This is because the main hurdle in
both algorithms is to compute a modular exponentiation in superposition.

In the case where the group order is of length $n$ bits and
the logarithm sought is of length $m \lll n$ bits bits, Ekerå's algorithm
exponentiates two elements to exponents of size $2m$ bits and $m$ bits
respectively. In Shor's algorithm, both exponents are instead of size $n \ggg m$ bits.

This difference is important since it is seemingly hard to build and operate large and
complex quantum computers. If the complexity of a quantum algorithm may be
reduced, in terms of the requirements that it imposes on the quantum computer,
this may well mean the difference between being able to execute the algorithm
and not being able to execute the algorithm.

\subsection{Our contributions in this paper}
In this paper, we generalize the algorithm of Ekerå for computing
short discrete logarithms by considering the setting where the quantum algorithm
is executed multiple times to yield multiple partial results.  This
enables us to further reduce the size of the exponent to only slightly more
than $m$ bits.

We then combine these partial results using
lattice-based techniques in a classical post-processing stage to yield the discrete
logarithm.
This allows for tradeoffs to be made between the number of times that the algorithm need be executed
on the one hand,
and the complexity of the algorithm and the requirements that it imposes on the
quantum computer on the other hand.

Furthermore, we describe applications of algorithms for computing short
discrete logarithms. In particular, we show how other important problems such as
those of factoring RSA integers and of finding the order of groups under
side information may be recast as short discrete logarithm problems.
By RSA integer we mean an integer that is the product of
two primes of similar size.

This immediately gives rise to an algorithm for factoring RSA integers that is
less complex than Shor's original general factoring algorithm in terms of the
requirements that it imposes on the quantum computer.

When factoring an $n$ bit integer using Shor's algorithm an exponentiation
is performed to an exponent of length $2n$ bits. In our algorithm, the exponent
is instead of length $(\frac 12 +\frac 1s) n$ bits where $s \ge 1$ is a
parameter that may assume any integer value.
As we remarked in the previous section, this reduction in complexity may well mean
the difference between being able to execute and not being able to execute the
algorithm.

\subsection{Overview of this paper}
In section \ref{section-notation} below we introduce some notation
and in section \ref{section-quantum-computing} we provide
a brief introduction to quantum computing.

In section \ref{section-computing-short-discrete-logarithms} we proceed to
describe our generalized algorithm for computing short discrete logarithms and
in section \ref{section-applications} we discuss interesting applications
for our algorithm and develop a factoring algorithm for RSA integers.
We conclude the paper and summarize our results in section \ref{section-summary-conclusion}.

\section{Notation}
\label{section-notation}
In this section, we introduce some notation used throughout this paper.

\begin{itemize}
  \item $u \text{ mod } n$ denotes $u$ reduced modulo $n$ and constrained to the
  interval
  \begin{align*}
    0 \le u \text{ mod } n < n.
  \end{align*}

  \item $\{ u \}_n$ denotes $u$ reduced modulo $n$ and constrained to the
  interval
  \begin{align*}
    -n/2 \le \{ u \}_n < n/2.
  \end{align*}

  \item $\left| \, a + \imag b \, \right| = \sqrt{a^2 + b^2}$ where
  $a, b \in \mathbb R$ denotes the Euclidean norm of $a + \imag b$ which is
  equivalent to the absolute value of $a$ when $b$ is zero.

  \item If $\vec u = \left( u_0, \: \ldots, \: u_{n-1} \right) \in \mathbb R^n$ is a
  vector then
  \begin{align*}
    \left| \, \vec u \, \right| = \sqrt{ \, u_0^2 + \ldots + u_{n-1}^2 \,}
  \end{align*}
  denotes the Euclidean norm of $\vec u$.
\end{itemize}

\section{Quantum computing}
\label{section-quantum-computing}
In this section, we provide a brief introduction to quantum computing.
The contents of this section is to some extent a layman's description of
quantum computing, in that it may leave out or overly simplify important details.

There is much more to be said on the topic of quantum computing. However, such
elaborations are beyond the scope of this paper. For more information,
the reader is instead referred to \cite{hirvensalo}. The extended
paper  \cite{shor1997} by Shor also contains a very good introduction and many
references to the literature.

\subsection{Quantum systems}
In a classical electronic computer, a register that consists of $n$ bits may
assume any one of $2^n$ distinct states $j$ for $0 \le j < 2^n$. The current
state of the register may be observed at any time by reading the register.

In a quantum computer, information is represented using qubits; not bits. A
register of $n$ qubits may be in a superposition of $2^n$ distinct states. Each
state is denoted $\ket{j}$ for $0 \le j < 2^n$ and a superposition of states,
often referred to as a quantum system, is written as a sum
\begin{align*}
	\ket{\Psi} = \sum_{j \, = \, 0}^{2^n - 1} c_j \ket{j}
	\quad \text{ where } \quad
	c_j \in \mathbb C
	\quad \text{ and } \quad
	\sum_{j \, = \, 0}^{2^n - 1} | \, c_j \,|^2 = 1
\end{align*}
that we shall refer to as the system function.

Each complex amplitude $c_j$ may be written on the form
$c_j = a_j e^{i \theta_j}$, where $a_j \in \mathbb R$ is a non-negative real
amplitude and $0 \le \theta_j < 2 \pi$ is a phase, so the system function may
equivalently be written on the form
\begin{align*}
	\ket{\Psi} = \sum_{j \, = \, 0}^{2^n - 1} a_j e^{\imag \theta_j} \ket{j}
	\quad\quad \text{ where } \quad\quad
	\sum_{j \, = \, 0}^{2^n - 1} a_j^2 = 1.
\end{align*}
\subsection{Measurements}
Similar to reading a register in a classical computer, the qubits in a register
may be observed by measuring the quantum system.

The result of such a measurement is to collapse the quantum system, and hence
the system function, to a distinct state. The probability of the system function
$\ket{\Psi}$ collapsing to $\ket{j}$ is $| \, c_j \, |^2 = a_j^2$.

\subsection{Quantum circuits}
It is possible to operate on the qubits that make up a quantum system using
quantum circuits. Such circuits are not entirely dissimilar from the electrical
circuits used to perform operations on bit registers in classical computers.

Given a quantum system in some known initial state, the purpose of a quantum
circuit is to amplify the amplitudes of a set of desired states, and to
suppress the amplitudes of all other states, so that when the system is
observed, the probability is large that it will collapse to a desired
state.

\subsection{The quantum Fourier transform}
In Shor's algorithms, that are the focus of this paper, the discrete quantum
Fourier transform (QFT) is used to to achieve amplitude amplification by
means of constructive interference.

The QFT maps each state in an $n$ qubit register to
\begin{align*}
\ket{j}
\quad \xrightarrow{\text{QFT}} \quad
\frac{1}{\sqrt{2^n}} \sum_{k \, = \, 0}^{2^n-1} e^{2 \pi \imag \, jk/2^n} \ket{k}
\end{align*}
so the QFT maps the system function
\begin{align*}
\ket{\Psi} = \sum_{j \, = \, 0}^{2^n-1} c_j \ket{j}
\quad \xrightarrow{\text{QFT}} \quad
\frac{1}{\sqrt{2^n}} \sum_{j \, = \, 0}^{2^n-1} \sum_{k \, = \, 0}^{2^n-1} c_j \, e^{2 \pi \imag \, jk/2^n} \ket{k}.
\end{align*}
\noindent

\subsubsection{Constructive interference}
If the above system is observed, the probability of it collapsing to $k$ is
\begin{align*}
\frac{1}{2^n}
\cdot
\left| \,
\sum_{j \, = \, 0}^{2^n-1} c_j \, e^{2 \pi \imag \, jk/2^n}
\, \right|^2.
\end{align*}

\noindent
Perceive the terms in the sum as vectors in $\mathbb{C}$.
If the vectors are
point in approximately the same direction, then the norm of their sum is likely
to be great giving rise to a large probability. For $k$ such
that this is indeed the case, constructive interference is said to arise.

The claim below summarizes the notion of constructive interference that
we use in this paper.

\label{section-constructive-interference}
\begin{claim}
\label{claim-sum-unit-vectors}
Let $\theta_j$ for $0 \le j < N$ be phase angles such that
$| \, \theta_j \, | \le \frac{\pi}{4}$. Then
\begin{align*}
	\left| \, \sum_{j \, = \, 0}^{N - 1} \e^{\imag \theta_j} \, \right|^2
  \ge
  \frac{N^2}{2}.
\end{align*}
\begin{proof}
\begin{align*}
\left| \, \sum_{j \, = \, 0}^{N - 1} \e^{\imag \theta_j} \, \right|^2
= \,
\left| \, \sum_{j \, = \, 0}^{N - 1} \left(
  \cos \theta_j + \imag \sin \theta_j
\right) \, \right|^2
\ge \,
\left| \, \sum_{j \, = \, 0}^{N - 1} \cos \theta_j \, \right|^2
\ge
\frac{N^2}{2}
\end{align*}
since for $j$ on the interval $0 \le j < N$ we have $| \, \theta_j \, | \le \frac{\pi}{4}$ which implies
\begin{align*}
  \frac{1}{\sqrt{2}} \le \cos \theta_j \le 1
\end{align*}
and the claim follows.
\end{proof}
\end{claim}

\section{Computing short discrete logarithms}
\label{section-computing-short-discrete-logarithms}
In this section, we describe a generalization of the algorithm for computing short discrete
logarithms previously introduced by Ekerå \cite{ekera}.

To describe the algorithm we
first formally define the discrete logarithm problem and introduce some
associated notation.

\subsection{The discrete logarithm problem}
Let $\mathbb G$ under $\odot$ be a group of order $r$ generated by $g$, and let
\begin{align*}
x = [d] \, g = \underbrace{g \odot g \odot \cdots \odot g \odot g}_{d \text{ times}}.
\end{align*}
Given $x$, a generator $g$ and a description of $\mathbb G$ and $\odot$ the
discrete logarithm problem is to compute $d = \log_g x$.

The bracket notation that we have introduced above is commonly used in the
literature to denote repeated application of the group operation regardless of
whether the group is written multiplicatively or additively.

\subsection{Algorithm overview}
The generalized algorithm for computing short discrete logarithms consists of
two stages; an initial quantum stage and a classical post-processing stage.

The initial quantum stage is described in terms of a quantum algorithm, see section
\ref{section-quantum-part}, that upon input of $g$ and $x = [d] \, g$ yields
a pair $(j, k)$.
The classical post-processing stage is described in terms of a classical algorithm, see
section \ref{section-classical-part}, that upon input of $s \ge 1$ ``good''
pairs computes and returns $d$.

The parameter $s$ determines the number of good pairs $(j, k)$ required to
successfully compute $d$. It furthermore controls the sizes of the index
registers in the algorithm, and thereby the complexity of
executing the algorithm on a quantum computer and the sizes of and amount
of information on $d$ contained in the two components $j,k$ of each pair.

In the special case where $s = 1$ the generalized algorithm is identical to the
algorithm in \cite{ekera}. A single good pair then suffices to compute $d$.

By allowing $s$ to be increased, the generalized
algorithm enables a tradeoff to be made between the requirements imposed by the algorithm on
the quantum computer on the one hand, and the number of times it needs to
be executed and the complexity of the classical
post-processing stage on the other hand.

We think of $s$ as a small constant. Thus, when we analyze the complexity
of the algorithm, and
in particular the parts of the algorithm that are executed classically, we
can neglect constants that depend on $s$.

\subsection{The quantum algorithm}
\label{section-quantum-part}
Let $m$ be the smallest integer such that $0 < d < 2^m$ and let
$\ell$ be an integer close to  $m/s$.
Provided that the order $r$ of $g$ is at least $2^{\ell+m} + 2^{\ell} d$,
the quantum algorithm described in this section will upon input
of $g$ and $x = [d] \, g$ compute and output a pair $(j, k)$.

A set of such pairs is then input to the classical algorithm to recover $d$.

\begin{enumerate}
  \item Let
  \begin{align*}
  \ket{\Psi} = \frac{1}{\sqrt{2^{2\ell+m}}}
  \sum_{a \, = \, 0}^{2^{\ell+m}-1}
  \sum_{b \, = \, 0}^{2^{\ell}-1}
  \ket{a} \ket{b} \ket{0}.
  \end{align*}
  where the first and second registers are of length $\ell+m$ and $\ell$ qubits.

  \item
  Compute $[a] \, g \odot [-b] \, x$ and store the result in the third register
  \begin{align*}
  \ket{\Psi} &=
  \frac{1}{\sqrt{2^{2\ell + m}}}
  \sum_{a \, = \, 0}^{2^{\ell+m}-1}
  \sum_{b \, = \, 0}^{2^{\ell}-1}
  \ket{a, b, [a] \, g \odot [-b] \, x} \\
  &=
  \frac{1}{\sqrt{2^{2\ell+m}}}
  \sum_{a \, = \, 0}^{2^{\ell+m}-1}
  \sum_{b \, = \, 0}^{2^{\ell}-1}
  \ket{a, b, [a - bd] \, g}.
  \end{align*}

  \item
  Compute a QFT of size $2^{\ell+m}$ of the first register and a QFT of size
  $2^{\ell}$ of the second register to obtain
  \begin{align*}
  	\ket{\Psi} =
  	\:&\frac{1}{\sqrt{2^{2\ell+m}}}
    \sum_{a \, = \, 0}^{2^{\ell+m}-1}
    \sum_{b \, = \, 0}^{2^{\ell}-1}
    \ket{a, b, [a - b d] \, g}
  	\quad \xrightarrow{\text{QFT}} \quad \\
  	&\frac{1}{2^{2\ell+m}}
  	\sum_{a, \, j \, = \, 0}^{2^{\ell+m}-1}
  	\sum_{b, \, k \, = \, 0}^{2^{\ell}-1}
  	\e^{\, 2 \pi \imag \, (a j + 2^m b k) / 2^{\ell+m}}
  	\ket{j, k, [a \,  - \, b d] \, g}.
  \end{align*}

  \item
  Observe the system in a measurement to obtain $(j, k)$ and $[e] \, g$.
\end{enumerate}

\subsubsection{Analysis of the probability distribution}
\label{section-small-d-analysis-probability}
When the system above is observed, the state $\ket{j, k, [e] \, g}$, where
$e = a - bd$, is obtained with probability
\begin{align*}
  \frac{1}{2^{2(2\ell+m)}} \: \cdot \,
  \left| \:
  \sum_a
  \sum_b
  \exp \left[ \frac{2 \pi \imag}{2^{\ell+m}} \, (a j +  2^{m} \, b k) \, \right]
  \: \right|^2
\end{align*}
where the sum is over all pairs $(a, b)$ that produce this specific $e$.
Note that the assumptions that the order
$r \ge 2^{\ell+m} + 2^{\ell} d$
imply that no reduction modulo $r$ occurs when $e$ is computed.

In what follows, we re-write the above expression for the probability on a form
that is easier to use in practice.

\begin{enumerate}
  \item Since $e = a - bd$ we have $a = e + bd$ so the probability may be written
  \begin{align*}
    \frac{1}{2^{2(2\ell+m)}} \: \cdot \,
    \left| \:
    \sum_b
    \exp \left[ \frac{2 \pi \imag}{2^{\ell+m}} \, ((e + bd) j +  2^{m} b k) \, \right]
    \: \right|^2.
  \end{align*}
  where the sum is over all $b$ in
  $\{\, 0 \le b < 2^{\ell} \:|\: 0 \le a = e + bd < 2^{\ell+m} \,\}$.

  \item Extracting the term containing $e$ yields
  \begin{align*}
    \frac{1}{2^{2(2\ell+m)}} \: \cdot \,
    \left| \:
    \sum_b
    \exp \left[ \frac{2 \pi \imag}{2^{\ell+m}} \, b(dj +  2^{m} k) \, \right]
    \: \right|^2.
  \end{align*}

  \item Centering $b$ around zero yields
  \begin{align*}
    \frac{1}{2^{2(2\ell+m)}} \: \cdot \,
    \left| \:
    \sum_b
    \exp \left[ \frac{2 \pi \imag}{2^{\ell+m}} \, (b - 2^{\ell-1}) (dj +  2^{m} k) \, \right]
    \: \right|^2.
  \end{align*}

  \item This probability may be written
  \begin{align*}
    \frac{1}{2^{2(2\ell+m)}} \: \cdot \,
    \left| \:
    \sum_b
    \exp \left[ \frac{2 \pi \imag}{2^{\ell+m}} \, (b - 2^{\ell-1}) \{ dj +  2^{m} k \}_{2^{\ell+m}} \, \right]
    \: \right|^2.
  \end{align*}
  since adding or subtracting multiples of $2^{\ell+m}$ has no effect; it is equivalent
  to shifting the phase angle by a multiple of $2 \pi$.
\end{enumerate}

\subsubsection{The notion of a good pair $(j, k)$}
\label{section-notion-good-pair}
By claim \ref{claim-sum-unit-vectors} the sum above is large when
$| \, \{ dj + 2^m k \}_{2^{\ell+m}} \, | \le 2^{m - 2}$ since this condition implies
that the angle is less than or equal to $\pi / 4$.

This observation serves as our motivation for introducing the below notion
of a good pair, and for proceeding in the following sections to lower-bound the
number of good pairs and the probability of obtaining any specific good pair.

\begin{definition}
\label{definition-good-pair}
A pair $(j, k)$, where $j$ is an integer such that $0 \le j < 2^{\ell+m}$ and
is said to be {\em good} if
\begin{align*}
	\left| \, \{ dj + 2^m k \}_{2^{\ell+m}} \, \right| \le 2^{m-2}.
\end{align*}
\end{definition}
\noindent
Note that $j$ uniquely defines $k$ as $k$ gives
the $\ell$ high order bits of $dj$ modulo $2^{\ell+m}$.

\subsubsection{Lower-bounding the number of good pairs $(j, k)$}
\begin{lemma}
\label{lemma-count-good-pairs}
There are at least $2^{\ell+m-1}$ different $j$ such that
there is a $k$ such that $(j,k)$ is a good pair.
\end{lemma}

\begin{proof}
For a good pair
\begin{eqnarray} \label{eq:dj}
\left| \, \{ dj + 2^m k \}_{2^{\ell+m}} \, \right| =
  \left| \, \{ dj \}_{2^{m}} \, \right|
  \, \le \,
  2^{m-2}
\end{eqnarray}
and for each $j$ that satisfies (\ref{eq:dj}) there
is a unique $k$ such that $(j,k)$ is good.

Let $2^\kappa$ be the greatest power of two that divides $d$.
Since $0 < d < 2^m$ it must be that $\kappa \le m-1$.
As $j$ runs through all integers $0 \le j < 2^{\ell+m}$, the function $dj \text{ mod } 2^{m}$ assumes the value of each multiple of $2^\kappa$ exactly $2^{\ell+\kappa}$ times.

\item
Assume that $\kappa = m - 1$. Then the only possible values are $0$ and
$2^{m-1}$. Only zero
gives rise to a good pair. With multiplicity there are
$2^{\ell + \kappa} = 2^{\ell+m - 1}$
integers $j$ such that $(j, k)$ is a good pair.

\item
Assume that $\kappa < m-1$.
Then only the $2 \cdot 2^{m-\kappa-2}+1$ values
congruent to values on $[-2^{m-2}, 2^{m-2}]$
are such that $| \, \{ dj \}_{2^{m}} \,| \le 2^{m-2}$. With
multiplicity $2^{\ell + \kappa}$ there are
$2^{\ell + \kappa} \cdot (2 \cdot 2^{m-\kappa-2} + 1) \ge 2^{\ell+m - 1}$
integers $j$ such that $(j, k)$ is a good pair.

\item
In both cases there are at least $2^{\ell+m - 1}$ good pairs and so the lemma follows.
\end{proof}

\subsubsection{Lower-bounding the probability of a good pair $(j, k)$}
To lower-bound the probability of a good pair we first need to lower-bound
the number of pairs $(a, b)$ that yield a certain $e$.

\begin{definition}
\label{definition-Te}
Let $T_e$ denote the number of pairs $(a, b)$ such that
\begin{align*}
e = a - bd
\end{align*}
where $a, b$ are integers on the intervals $0 \le a < 2^{\ell+m}$ and $0 \le b < 2^\ell$.
\end{definition}

\begin{claim}
\label{claim-interval-e}
  \begin{align*}
  | \, e = a - bd \, | < 2^{\ell+m}
  \end{align*}
\begin{proof}
  The claim follows from $0 \le a < 2^{\ell+m}$, $0 \le b < 2^\ell$ and
  $d < 2^{m}$.
\end{proof}
\end{claim}

\noindent
\begin{claim}
\label{claim-sum-Te}
  \begin{align*}
    \sum_{e \, = \, -2^{\ell+m}}^{2^{\ell+m}-1} T_e = 2^{2 \ell + m}.
  \end{align*}
\begin{proof}
Since $a, b$ may independently assume $2^{\ell+m}$ and $2^\ell$ values, there
 are $2^{2\ell + m}$ distinct pairs $(a, b)$. From
this fact and claim \ref{claim-interval-e} the claim follows.
\end{proof}
\end{claim}

\begin{claim}
\label{claim-sum-Te-2}
\begin{align*}
\sum_{e \, = \, -2^{\ell+m}}^{2^{\ell+m}-1} T_e^2 \ge 2^{3\ell + m - 1}.
\end{align*}
\begin{proof}
The claim follows from the Cauchy–Schwarz inequality and claim \ref{claim-sum-Te} since
\begin{align*}
2^{2(2\ell+m)} =
\left( \sum_{e \, = \, -2^{\ell+m}}^{2^{\ell+m}-1} T_e \right)^2
\le\,
\left( \sum_{e \, = \, -2^{\ell+m}}^{2^{\ell+m}-1} 1^2 \right)
\left( \sum_{e \, = \, -2^{\ell+m}}^{2^{\ell+m}-1} T_e^2 \right).
\end{align*}
\end{proof}
\end{claim}

\noindent
We are now ready to demonstrate a lower-bound on the probability of obtaining a
good pair using the above definition and claims.

\begin{lemma}
\label{lemma-probability-good-pair}
The probability of obtaining any specific good pair $(j, k)$ from a single
execution of the algorithm in section \ref{section-quantum-part} is at least
$2^{-m-\ell-2}$.
\end{lemma}

\begin{proof}
For a good pair
\begin{align*}
	\left| \, \frac{2 \pi}{2^{\ell+m}} \, (b - 2^{\ell-1}) \, \{ dj + 2^m k \}_{2^{\ell+m}} \, \right|
  \le \frac{2 \pi}{2^{\ell+2}} \, \left| \, b - 2^{\ell-1} \, \right |
  \le \frac{\pi}{4}
\end{align*}
for any integer $b$ on the interval $0 \le b < 2^\ell$. It therefore follows
from claim \ref{claim-sum-unit-vectors}
that the probability of observing $(j, k)$ and $[e] \, g$
is at least
\begin{align*}
	\frac{1}{2^{2(2\ell+m)}} \cdot
	\left| \:
	\sum_{b}
	\exp \left[ \frac{2 \pi \imag}{2^{2 \ell}} \, (b - 2^{\ell-1}) \, \{ dj + 2^m k \}_{2^{\ell+m}} \right]
	\: \right|^2 \ge \frac{T_e^2}{2 \cdot 2^{2 (2\ell+m)}}
\end{align*}
Summing this over all $e$ and using claim \ref{claim-sum-Te-2} yields
\begin{align*}
\sum_{e \, = \, -2^{\ell+m}}^{2^{\ell+m}-1} \frac{T_e^2}{2 \cdot 2^{2 (2\ell+m)}} \ge 2^{-m-\ell-2}
\end{align*}
from which the lemma follows.
\end{proof}
We note that by Lemma~\ref{lemma-count-good-pairs} and
Lemma~\ref{lemma-probability-good-pair} the probability of the algorithm
yielding a good pair as a result of a single execution is at least $2^{-3}$.
In the next section, we describe how to compute $d$ from a set of $s$
good pairs.

\subsection{Computing $d$ from a set of $s$ good pairs}
\label{section-classical-part}
In this section, we specify a classical algorithm that upon input
of a set of $s$ distinct good pairs $\{ (j_1, k_1), \: \dots, \: (j_s, k_s) \}$, that result
from multiple executions of the algorithm in section \ref{section-quantum-part},
computes and outputs $d$.

The algorithm uses lattice-based techniques. To
introduce the algorithm, we first need to define the lattice $L$.

\begin{definition}
\label{definition-L}
Let $L$ be the integer lattice generated by the row span of
\begin{align*}
\left[
\begin{array}{ccccc}
 j_1 &  j_2 & \cdots &  j_s & 1 \\
2^{\ell+m} & 0 & \cdots & 0 & 0 \\
0 & 2^{\ell+m} & \cdots & 0 & 0 \\
\vdots & \vdots & \ddots & \vdots & \vdots \\
0 & 0 & \cdots & 2^{\ell+m} & 0
\end{array}
\right].
\end{align*}
\end{definition}

\noindent
The algorithm proceeds as follows to recover $d$ from $\{ (j_1, k_1), \: \dots, \: (j_s, k_s) \}$.
\begin{enumerate}
\item Let $\vec v = ( \,\{ -2^m k_1 \}_{2^{\ell+m}}, \: \dots, \:  \,\{ -2^m k_s \}_{2^{\ell+m}}, \: 0 ) \in \mathbb Z^{s+1}$.

For all vectors $\vec u \in L$ such that
\begin{align*}
  | \, \vec u - \vec v \, | < \sqrt{s/4+1} \cdot 2^m
\end{align*}
test if the last component of $\vec u$ is $d$. If so return $d$.

This test may be performed by checking if $x = [d] \, g$.

\item If $d$ is not found in step 1 or the search is infeasible the algorithm fails.
\end{enumerate}
As $s$ is a constant, all vectors close to $\vec v$ can be found efficiently.
We return to the problem of there possibly being many close vectors in
Lemma~\ref{lemma:lattice} below.

\subsubsection{Rationale and analysis}
\label{section-rationale-analysis}
For any $m_1, \: \dots, \: m_s \in \ZZ$ the vector
\begin{align*}
  \vec u = (\{ dj_1 \}_{2^{\ell+m}} + m_1 2^{\ell+m}, \: \dots, \: \{ dj_s \}_{2^{\ell+m}} + m_s 2^{\ell+m}, \, d ) \in L.
\end{align*}
The above algorithm performs an exhaustive search of all vectors in $L$ at
distance at most $\sqrt{s/4+1} \cdot 2^m$ from $\vec v$ to find $\vec u$ for some $m_1, \: \dots, \: m_s$. It then
recovers $d$ as the second component of $\vec u$. The search will succeed in finding $\vec u$ since
\begin{align*}
\left| \, \vec u - \vec v \, \right|
&= \sqrt{\: d^2 + \sum_{i \, = \, 1}^s \left( \{ dj_i \}_{2^{\ell+m}} + m_i \, 2^{\ell+m} - \{ -2^m k_i \}_{2^{\ell+m}} \right)^2} \\
&= \sqrt{\: d^2 + \sum_{i \, = \, 1}^s \left( \{ dj_i + 2^m k_i \}_{2^{\ell+m}} \right)^2} < \sqrt{s/4+1} \cdot 2^m
\end{align*}
since $0 < d < 2^m$ and $|\{ dj + 2^m k \}_{2^{\ell+m}}| \le 2^{m-2}$ by the definition
of a good pair, and since $m_1, \: \dots, \: m_s$ may be freely selected to obtain equality.

Whether the search is computationally feasible depends on the number of vectors
in $L$ that lie within distance $\sqrt{s/4+1} \cdot 2^m$ of $\vec v$. This
number is related to the norm of the shortest vector in the lattice.

Note that the determinant of $L$ is $2^{(\ell+m)s}\approx 2^{m(s+1)}$. As the lattice is
$s+1$-dimensional we would expect the shortest vector to be of length about $2^m$.
This is indeed true with high probability.

\begin{lemma}\label{lemma:lattice}
The probability that $L$ contains a vector $\vec u= (u_1, \ldots u_{s+1})$
with $| \, u_i \, | < 2^{m-3}$ for $1\leq i \leq s+1$ is bounded by $2^{-s-1}$.
\end{lemma}

\begin{proof}
Take any integer $\vec u$ with all coordinates strictly bounded by $2^{m-3}$.

If $2^\kappa$ is the largest
power of two that divides $u_{s+1}$ then  $u_i$ must also
be divisible by $2^{\kappa}$ for $\vec u$ to belong to any lattice in
the family. By family we mean all lattices on the same form and degree as $L$, see definition \ref{definition-L}.  If it this is true for all $i$ then  $\vec u$ belongs to $L$ for
$2^{s \kappa}$ different values of $(j_i)_{i=1}^s$.

There
are $2^{(m-2-\kappa) (s+1)}$  vectors $\vec u$ with all coordinates divisible
by $2^{ \kappa}$ and bounded in absolute value by $2^{m-3}$.
We conclude that the total number of lattices $L$ that contain
such a short vector is bounded by
$$
\sum_\kappa 2^{(m-2-\kappa) (s+1)} 2^{\kappa  s} \leq
2^{1+(m-2) (s+1)}.
$$
As the number of $s$ tuples of good $j$ is at least
$2^{s(\ell +m -1)}$, the lemma follows.
\end{proof}
Lemma~\ref{lemma:lattice} shows that with good probability
the number of lattice points that
$| \, \vec u - \vec v \, | < \sqrt{s/4+1}$ is a
constant that only depends on $s$ and thus we can efficiently
find all such vectors.

\subsection{Building a set of $s$ good pairs}
The probability of a single execution of the quantum
algorithm in section~\ref{section-quantum-part} yielding a good pair is
at least $2^{-3}$ by Lemma~\ref{lemma-count-good-pairs} and
Lemma~\ref{lemma-probability-good-pair}.
Hence, if we execute the quantum algorithm $t=8s$ times, we obtain
a set of $t$ pairs that we expect contains at least $s$ good pairs.

In theory, we may then recover $d$ by executing the classical algorithm in
section~\ref{section-classical-part} with respect to all ${{t}\choose{s}}$
subsets of $s$ pairs selected from this set. Since $s$ is a constant, this approach
implies a constant factor overhead in the classical part of the algorithm. It
does not affect the quantum part of the algorithm. We summarize these
ideas in Theorem \ref{theorem-main} below.

In practice, however, we suspect that it may be easier to recover $d$. First
of all, we remark that we have only established a lower bound on the probability that a good pair is
yielded by the algorithm. This bound is not tight and we expect the actual
probability to be higher than is indicated by the bound.

Secondly, we have only analyzed the probability of the classical algorithm
in section~\ref{section-classical-part} recovering $d$ under the assumption
that all $s$ pairs in the set input are good. It might however well turn out to be true
that the algorithm will succeed in recovering $d$ even if not all pairs in the input set are good.

\subsection{Main result}
In this subsection we summarize the above discussion in a main theorem. Again, we
stress that the approach outlined in the theorem is conservative.
\begin{theorem}
  \label{theorem-main}

  Let $d$ be an integer on $0 < d < 2^m$, let $s \ge 1$ be a fixed
   integer,
  let $\ell$ be an integer close to $m/s$ and let $g$ be a generator of a
  finite cyclic group of order $r \ge 2^{\ell+m} + 2^\ell d$.

  Then there exists a quantum algorithm that yields a pair as output
  when executed with $g$ and $x = [d] \,g$ as input.
  The main operation in this algorithm is an
  exponentiation of $g$ in superposition to an exponent of length $\ell + m$ bits.

  If this algorithm is executed $\mathcal O(s)$ times to yield a set of
  pairs $\mathcal S$, then there exists a polynomial time classical
  algorithm that computes $d$ if executed with
  all unordered subsets of $s$ pairs from $\mathcal S$ as input.
\end{theorem}

\noindent
The proof of Theorem \ref{theorem-main} follows from the above discussion.

Note that the order $r$ of the group need not be explicitly known. It suffices
that the above requirement on $r$ is met. Note furthermore
that it must be possible to implement the group operation efficiently on a quantum
computer.

\subsection{Implementation remarks}
We have described the above algorithm in terms of it using two index
registers.

Similarly, we have described the algorithm in terms of the quantum system being
initialized, of a quantum circuit then being executed and of the quantum system
finally being observed in a measurement. However, this is not necessarily the manner in which the
algorithm would be implemented in practice on a quantum computer.

For example, Mosca and Ekert \cite{mosca-ekert} have described optimizations of
Shor's general algorithm for computing discrete logarithm that allow the index
registers to be truncated. These optimizations,
alongside other optimizations of Shor's original algorithm for computing
discrete logarithms, may in many cases be applicable also to our algorithm for
computing short discrete logarithms. This is due to the fact that the quantum
stages are fairly similar.

Depending on the specific architecture of the
quantum computer on which the algorithm is to be implemented, it is likely that
different choices will have to be made with respect to how the implementation
is designed and optimized.

In this paper we therefore describe our algorithm in the simplest possible
manner, without taking any of these optimizations into account.

\section{Applications}
\label{section-applications}
In this section, we describe applications for the generalized algorithm for
computing short discrete logarithms introduced in the previous section.

\subsection{Computing short discrete logarithms}
Quantum algorithms for computing short discrete logarithms may be used to attack
certain instantiations of asymmetric cryptographic schemes that rely on the
computational intractability of this problem.

A concrete example of such an application is to attack Diffie-Hellman over
finite fields when safe prime groups are used in conjunction with short exponents.

The existence of efficient specialized algorithms for computing short discrete
logarithms on quantum computers should be taken into account when selecting
and comparing domain parameters for asymmetric cryptographic schemes that rely
on the computational intractability of the discrete logarithm problem.

For further details, the reader is referred to the extended rationale in
\cite{ekera} and to the references to the literature provided in that paper.

\subsection{Factoring RSA integers}
In this section we describe how the RSA integer factoring problem may be
recast as a short discrete logarithm problem by using ideas from
Håstad et al. \cite{hastad}, and the fact
that our algorithm does not require the group order to be known.

This immediately gives rise to an algorithm for
factoring RSA integers that imposes smaller requirements on the quantum
computer than Shor's general factoring algorithm.

\subsubsection{The RSA integer factoring problem}
Let $p$ and $q \neq p$ be two random odd primes such that $2^{n-1} < p, q < 2^n$.
The RSA integer factoring problem is then to factor $N=pq$ into $p$ and $q$.

The RSA integer factoring problem derives its name from Rivest, Shamir and Adleman who proposed
to base the widely deployed RSA cryptosystem on the computational intractability of
the RSA integer factoring problem.

\subsubsection{The factoring algorithm}
Consider the multiplicative group $\mathbb Z^*_{N}$ to the ring of
integers modulo $N$. This group has order $\phi(N) = (p-1)(q-1)$.
Let $\mathbb G$ be some cyclic subgroup to $\mathbb Z^*_{N}$.

Then  $\mathbb G$ has order $\phi(N) / t$ for some $t \mid \phi(N)$
such that $t \ge \gcd(p-1, q-1)$. In what follows below, we assume that $\phi(N) / t > (p + q - 2)/2$.

\begin{enumerate}
\item
\label{factoring-step-select-g}
Let $g$ be a generator of $\mathbb G$. Compute $x = g^{(N-1)/2}$. Then $x \equiv g^{(p+q-2)/2}$.

\item
\label{factoring-step-compute-dlog}
Compute the short discrete logarithm $d = (p + q - 2)/2$ from $g$ and $x$.

\item
Compute $p$ and $q$ by solving the quadratic equation
\begin{align*}
 N = (2d - q + 2)q = 2(d + 1)q - q^2
\end{align*}
where we use that $2d+2 = p + q$.   This yields
\begin{align*}
  p, q = c \pm \sqrt{c^2 - N}
  \quad \text{ where } \quad
  c = d + 1.
\end{align*}
We obtain $p$ or $q$ depending on the choice of sign.
\end{enumerate}
To understand why we obtain a short logarithm, note that
\begin{align*}
  N - 1 = pq - 1 = (p-1) + (q-1) + (p-1)(q-1)
\end{align*}
from which it follows that
$(N - 1)/2 \equiv (p + q - 2)/2 \text{ mod } \phi(N)/t$
provided that the above assumption that $\phi(N) / t > (p + q - 2)/2$ is met.

The only remaining difficulties are the selection of the generator in step \ref{factoring-step-select-g}
and the computation of the short discrete logarithm in step \ref{factoring-step-compute-dlog}.

\subsubsection{Selecting the generator $g$}
We may pick any cyclic subgroup
$\mathbb G$ to $\mathbb Z_{N}^*$ for as long as its order $\phi(N)/t$ is
sufficiently large. It suffices that
$\phi(N)/t > (p + q - 2)/2$ and that the discrete logarithm can be computed,
see section \ref{section-remarks-computing-short-discrete-logarithm} below
for more information.

This implies that we may simply select an element $g$ uniformly at random
on the interval $1 < g < N-1$ and use it as the generator in step
\ref{factoring-step-select-g}.

\subsubsection{Computing the short discrete logarithm}
\label{section-remarks-computing-short-discrete-logarithm}
To compute the short discrete logarithm in step \ref{factoring-step-compute-dlog},
we use the algorithm in section
\ref{section-computing-short-discrete-logarithms}.
This algorithm requires that the order
\begin{align*}
  \phi(N) / t \ge 2^{\ell + m} + 2^\ell d
  \quad\Rightarrow\quad
  \phi(N)/t \ge 2^{\ell + m + 1}
\end{align*}
where we have used that $0 < d < 2^m$. We note that
\begin{align*}
  2^{n} \le d = (p + q - 2)/2 < 2^{n+1}
  \quad\Rightarrow\quad
  m = n + 1.
\end{align*}
Furthermore, we note that $\phi(N) = (p-1)(q-1) \ge 2^{2(n-1)}$ which implies
\begin{align*}
\phi(N) / t \ge 2^{2(n-1)} / t \ge 2^{\ell+m+1} = 2^{\ell+n+2}
\quad\Rightarrow\quad
t < 2^{2(n-1)-(n + \ell + 2)} = 2^{n - \ell - 4}.
\end{align*}

Recall that $\ell = m / s = (n + 1)/s$ where $s \ge 1$. For
random $p$ and $q$, and a randomly
selected cyclic subgroup to $\mathbb Z_{N}^*$, the requirement
$t < 2^{n - \ell - 4}$ is hence met with overwhelming probability for any $s > 1$.

We remark that further optimizations are possible. For instance
the size of the logarithm may be reduced by computing $x = g^{(N-1)-2^{n}}$
since $p,q > 2^{n-1}$.

\subsubsection{Generalizations}
We note that the algorithm proposed in this section can be generalized.

In particular, we have assumed above that the two factors are of the same length
in bits as is the case for RSA integers. This requirement can be relaxed. As
long as the difference in length between the two factors is not too great, the
above algorithm will give rise a short discrete logarithm that may be computed
using our generalized algorithm in section \ref{section-computing-short-discrete-logarithms}.

\subsection{Order finding under side information}
In this section, we briefly consider the problem of computing the order of
a cyclic group $\mathbb G$ when a generator $g$ for the group is available
and when side information is available in the form of an estimate of the group order.

\subsubsection{The algorithm}
Let $\mathbb G$ be a cyclic group of order $r$. Let $r_0$ be a known
approximation of the order such that $0 \le r - r_0 < 2^m$. The problem of
computing the order $r$ under the side information $r_0$ may then be recast
as a short discrete logarithm problem:

\begin{enumerate}
\item Let $g$ be a generator of $\mathbb G$. Compute $x = g^{-r_0}$. Then $x \equiv g^{r - r_0}$.

\item Compute the short discrete logarithm $d = r - r_0$ from $g$ and $x$.

\item Compute the order $r = d + r_0$.
\end{enumerate}

\section{Summary and conclusion}
\label{section-summary-conclusion}
In this paper we have generalized the quantum algorithm for computing short discrete
logarithms previously introduced by Ekerå \cite{ekera} so as to allow for
various tradeoffs between the number of times that the algorithm need be executed
on the one hand,
and the complexity of the algorithm and the requirements it imposes on the
quantum computer on the other hand.

Furthermore, we have described applications for algorithms for computing short
discrete logarithms. In particular, we have shown how other important problems such as
those of factoring RSA integers and of finding the order of groups under
side information may be recast as short discrete logarithm problems.
This immediately gives rise to an algorithm for factoring RSA integers that
is less complex than Shor's general factoring algorithm in the sense that it
imposes smaller requirements on the quantum computer.

In both our algorithm and Shor's algorithm, the main hurdle is to compute
a modular exponentiation in superposition. When factoring an $n$ bit integer,
the exponent is of length $2n$ bits in Shor's algorithm, compared to
slightly more than $n/2$ bits in our algorithm. We have made essentially two
optimizations that give rise to this improvement.

First, we gain a factor of two by re-writing the factoring problem as a short discrete
logarithm problem and solving it using our algorithm for computing short
discrete logarithms. One
way to see this is that we know an approximation $N$ of the order
$\phi(N)$. This gives us a short discrete logarithm problem and our
algorithm for solving it does not require the order to be known beforehand.

Second, we gain a factor of two by executing
the quantum algorithm multiple times to yield a set of partial results. We then
recover the discrete logarithm $d$ from this set in a classical post-processing step. The
classical algorithm uses lattice-based techniques. It constructs a
lattice $L$ and a vector $\vec v$ from the set of partial results and recovers
 $d$ by exploring vectors in $L$ close to $\vec v$.

\subsection{Remarks on generalizations of these techniques}
The second optimization above may seemingly be generalized and applied to
other quantum algorithms such as for example
Shor's
algorithm for factoring general integers. This allows a factor two to be
gained.

We have not yet analyzed this case in detail but the idea is basically to exponentiate a random group element to an exponent of length
$\ell + m$ bits in superposition, where the order $r$ of the element is such that $r < 2^m$ and
$\ell = m / s$.

The quantum algorithm is executed multiple times to yield partial results in the form of integers $j_1, j_2, \dots$
that are such that $| \, \{ j_i \, r \}_{2^{\ell + m}} \,| \le 2^{m-2}$.

In the classical post-processing step lattice-based techniques are then used to
extract the order. The lattice $L$ is on the same form as in our algorithm, but
$\vec v$ is now the zero vector, and we hence seek a short non-zero vector in $L$.
The last component of this vector is $r$.

\section*{Acknowledgments}
Support for this work was provided by the Swedish NCSA, that
is a part of the Swedish Armed Forces, and by the Swedish Research Council (VR).


\begin{thebibliography}{99}
\bibitem{hirvensalo} M. Hirvensalo, \emph{\q{Quantum Computing}}, 2nd edition, Natural Computing Series, Springer Verlag, 2004.

\bibitem{ekera} M. Ekerå, \emph{\q{Modifying Shor's algorithm to compute short discrete logarithms}}, in IACR ePrint Archive, report \href{https://eprint.iacr.org/2016/1128}{2016/1128}, 2016.

\bibitem{hastad} J. Håstad, A. Schrift, A. Shamir, \emph{\q{The Discrete Logarithm Modulo a Composite Hides O(n) bits}}, in Journal of Computer and System Science, Vol 47, No 3, 1993, pp 376-404.

\bibitem{mosca-ekert} M. Mosca, A. Ekert,
\q{\emph{The Hidden Subgroup Problem and Eigenvalue Estimation on a Quantum Computer}},
in proceeding from the first NASA International Conference,
Quantum Computing and Quantum Communications, volume 1509, 1999, pp. 174–188.

\bibitem{shor1994} P. W. Shor, \emph{\q{Algorithms for Quantum Computation: Discrete Logarithms and Factoring}}, in proceeding from the 35th Annual Symposium on Foundations of Computer Science, Santa Fe, NM, November 20–22, 1994, IEEE Computer Society Press, pp. 124–134.

\bibitem{shor1997} P. W. Shor, \emph{\q{Polynomial-time algorithms for prime factorization and discrete logarithms on a quantum computer}}, in SIAM Journal of Computing, volume 26, no 5, 1997, pp. 1484-1509.
\end{thebibliography}
\end{document}